\documentclass[conference]{IEEEtran}
\IEEEoverridecommandlockouts
\usepackage{cite}
\usepackage{amsmath,amsthm,amssymb,amsfonts}
\usepackage{algorithmic}
\usepackage{graphicx}
\usepackage{textcomp}
\usepackage{xcolor}
\usepackage[margin=0.7in]{geometry}
\usepackage[caption=false,font=footnotesize]{subfig}
\DeclareMathOperator{\Tr}{Tr}
\graphicspath{{figures/drawings/}{figures/pdf/}}
\begin{document}

\title{UAV Swarms as Amplify-and-Forward MIMO Relays\\
\thanks{This work was supported in part by the CONIX Research Center, one of six centers in JUMP, a Semiconductor Research Corporation (SRC) program sponsored by DARPA.}
}

\author{\IEEEauthorblockN{Samer Hanna, Enes Krijestorac, Han Yan, and Danijela Cabric}
\IEEEauthorblockA{\textit{Electrical and Computer Engineering Department,} \\ \textit{University of California, Los Angeles} \\ 
		samerhanna@ucla.edu, enesk@g.ucla.edu, yhaddint@ucla.edu, danijela@ee.ucla.edu 
}
}

\maketitle

\begin{abstract}
Unmanned aerial vehicles provide new opportunities for performance improvements in future  wireless communications systems. For example, they can act as relays that extend the range of  a communication link and improve the capacity. 
Unlike conventional relays that are deployed at fixed locations, UAVs can change their positions to optimize the capacity or range on demand.  In this paper,  we consider using a swarm of UAVs as amplify-and-forward MIMO relays to provide connectivity between an obstructed multi-antenna equipped source and destination. 
We start by optimizing UAV placement for the single antenna case, and analyze its dependence on the noise introduced by the relay, its gain, and transmit power constraint.  We extend our analysis for an arbitrary UAV swarm and show how the MIMO link capacity can be optimized by changing the distance of the swarm to the source and the destination.
Then, we consider the effect of optimizing the positions of the UAVs within the swarm and derive an upper bound for the capacity at  any given placement of the swarm. 
We also propose a simple near optimal approach to find the positions that optimize the capacity for the end-to-end link given that the source and the destination have uniform rectangular arrays.
\end{abstract}


\section{Introduction}
Due to their mobility and low cost, unmanned aerial vehicles (UAVs) have found their way to many applications in recent years. Examples  include package delivery, law enforcement, search and rescue, etc. Driven by this demand, UAVs are expected to become more prevalent, which will further drive the development of this technology and demand regulations that will allow for higher presence of UAVs in the low-altitude air space.
Following this trend, UAVs are getting an increased attention in the telecommunications sector due to the multitude of opportunities they can provide  \cite{zeng_opportunities_2016}. Using UAVs has recently emerged as an idea to respond to high localized traffic demands in next generation cellular networks \cite{li2017uav, wu2018uav, mozaffari2016efficient}. Beyond using UAVs as basestations, UAVs can be used as relays to extend the range of communication, boost capacity, or as a substitute for failed infrastructure in case of disasters.

Wireless relaying is one of the classical ways to improve data rates, while increasing reliability by combating shadowing. By using more than one relay along with multiple antennas at the source and the destination,  multiple-input multiple-output (MIMO) relay networks are able to boost the capacity ~\cite{capacity_scaling_2006}.
Traditionally, relaying approaches relied on using fixed relays. However, fixed relays, typically deployed on the ground, are unable to meet fluctuating demands or respond to failures in communications infrastructure.  Deploying  relays on UAVs provides new opportunity to exploit agility of motion of UAVs and the capacity improvements offered by wireless relays.
  
  UAVs typically fly at altitudes of a few hundred feet above the ground, which could provide a line-of-sight (LOS) channel between an obstructed source and destination. By using multiple UAVs along with multiple antennas at the source and the destination, MIMO capacity gains can be leveraged. The MIMO capacity of this end-to-end link depends on the placement of UAVs, as they directly affect both the source-to-UAV and the UAV-to-destination channels. Hence, by optimizing the placement of individual UAVs within the swarm,  capacity can be improved.

There is a significant interest in using UAVs as relays in recent literature. In~\cite{chen_optimum_2018}, the authors considered the optimal placement of a UAV relay that minimizes the outage probability. Other works have addressed the problem of finding the optimal trajectory and transmit power of a mobile relay~\cite{zeng_throughput_2016, jiang_power_2018}. In~\cite{larsen_optimal_2017}, the authors show that positioning a UAV asymmetrically between two ground nodes could result in better service than with a UAV placed at the center position when using stepwise adaptive modulation. 
All these works have only considered a single UAV relay and do not apply to MIMO links.
 In~\cite{hanna2018distributed},  algorithms for optimizing the placement of a UAV swarm were developed, but here only a single-hop link was considered.   The placement of multiple UAVs as relays in double and multiple hop relay networks was optimized in \cite{chen_multiple_2018}. However, this work  assumes that only a single UAV is transmitting at a time, thus full MIMO gains are not leveraged. 
 In~\cite{kalogerias2018spatially}, the placement of multiple UAV relays in a dynamic channel is analyzed but the source and destination are assumed to have only single antennas.

In this work, we study the effect of changing the positions and the arrangement of a UAV swarm acting as amplify-and-forward relay cluster between a multiple antenna transmitter and receiver under an obstructed direct link. We also propose a method for optimizing the channel capacity by controlling the placement of the UAV relays. We start by considering the single antenna case and show how the UAV relay design parameters such as amplification gain and noise figure affect its position for maximized capacity. We derive an upper bound for the achievable capacity in case of a UAV swarm, and show that the single antenna analysis can be extended to the UAV swarm case. For the swarm, we demonstrate the gains of optimizing positions of the UAVs within the swarm and we propose a simple approach to find the positions that can attain the upper bound for a transmitter and receiver consisting of a uniform rectangular array for some separations between the swarm and the transmitter, while giving a better capacity than random placement on the average for all separations.

The rest of the paper is organized as follows. We start by defining the system model used throughout this work in Section \ref{sec:system_model}. An analysis for the capacity that can be obtained by a single UAV or a UAV swarm is presented in Section \ref{sec:cap}, while  Section \ref{sec:method} proposes a method to find this placement for URA source and destination. Simulation results analyzing impact of the UAV swarm placement and configuration are presented in Section \ref{sec:simulations}. Section \ref{sec:conclusion} concludes the paper and presents directions for future research.

\section{System Model} \label{sec:system_model}
\renewcommand{\b}[1]{\boldsymbol{\mathrm{#1}}}
\newcommand{\br}[1]{\bar{#1}}

\newcommand{\C}[1]{\mathbb{C}^{#1}}
\newcommand{\mSymbTransmitter}{\mathrm{T}}
\newcommand{\mSymbUAV}{\mathrm{U}}
\newcommand{\mSymbReceiver}{\mathrm{R}}
\newcommand{\mNt}{N_\mSymbTransmitter}
\newcommand{\mPt}{P_\mSymbTransmitter}
\newcommand{\mNr}{N_\mSymbReceiver}

\newcommand{\mNu}{N_\mSymbUAV}
\newcommand{\mPu}{P_\mSymbUAV}

\newcommand{\mH}{\b{H}}
\newcommand{\mHI}{\b{H}_1}
\newcommand{\mHII}{\b{H}_2}
\newcommand{\mHbI}{\br{\b{H}}_1}
\newcommand{\mHbII}{\br{\b{H}}_2}
\newcommand{\mHt}{\Tilde{\b{H}}}

\newcommand{\mFb}[1]{ \|#1\|_F }
\newcommand{\mlHt}[1]{\hat{\psi}_{#1}}
\newcommand{\mlHI}{\psi_{1}}
\newcommand{\mlHII}{\psi_{2}}

\newcommand{\mh}{h}
\newcommand{\mhI}{h_1}
\newcommand{\mhII}{h_2}

\newcommand{\ms}{\b{s}}
\newcommand{\mt}[1]{\b{x}_{#1}}
\newcommand{\mr}[1]{\b{y}_{#1}}
\newcommand{\mn}[1]{\b{n}_{#1}}
\newcommand{\mnVar}[1]{\b{\Sigma}_{#1}}
\newcommand{\mnvar}[1]{\sigma^2_{#1}}
\newcommand{\mNF}{f_{\mSymbUAV}}

\newcommand{\msSc}{s}
\newcommand{\mtSc}[1]{x_{#1}}
\newcommand{\mrSc}[1]{y_{#1}}
\newcommand{\mnSc}[1]{n_{#1}}

\newcommand{\mD}{\b{D}}
\newcommand{\mAlpha}[1]{\alpha_{#1}}

\newcommand{\mPMax}[1]{P^{\max}_{#1}}
\newcommand{\mAlphaMax}[1]{\alpha^{\max}_{#1}}

\newcommand{\mDist}{R}
\newcommand{\mDistI}{R_1}

\newcommand{\mf}{f}
\newcommand{\mg}{g}

\begin{figure}[t!]
	\centering
	\includegraphics[width=0.45\textwidth]{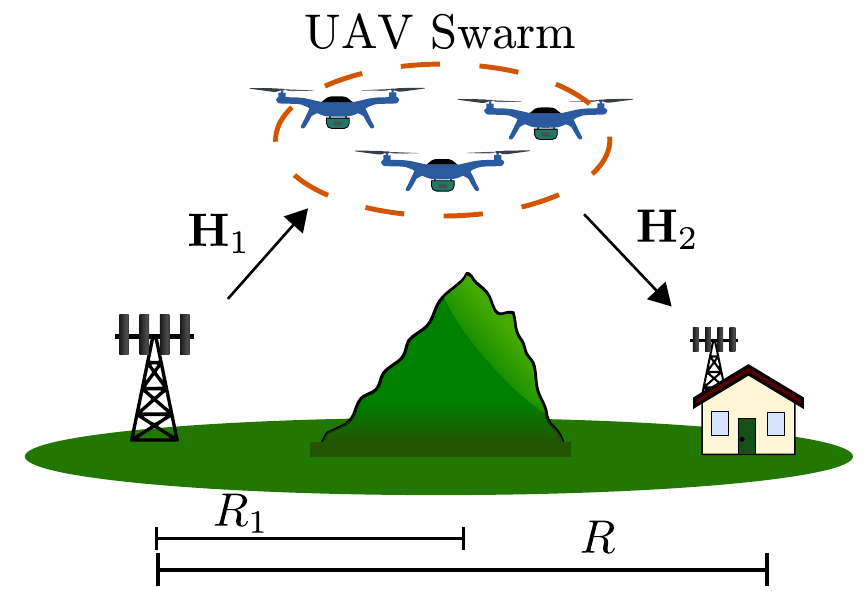}
	\caption{UAV swarm assisting obstructed MIMO link.}
	\label{fig:setup}
\end{figure}

We consider a link between a transmitter with $ \mNt $ antennas and a receiver with $ \mNr $ antennas separated by a distance $ \mDist $. We assume that direct communication between the transmitter and the receiver is not possible due to obstructions. We use $\mNu$ UAVs placed at distance $\mDistI$ from the transmitter, each equipped with a single-antenna, to enable this link and maximize its capacity. Each UAV acts as an amplify-and-forward  relay which simply receives the signal, amplifies, and re-transmits it in a synchronized manner.

The channel between the transmitter and the UAVs is denoted by $ \mHI \in \C{\mNu \times \mNt} $, while the channel between the UAVs and the receiver is denoted by $ \mHII \in \C{\mNr \times \mNu} $. We assume that both channels are strong line-of-sight (LOS) channels where each element is given by
$
	\left[\mH\right]_{i,j} = \frac{\lambda}{4 \pi d_{i,j}}  e^{\frac{j 2 \pi d_{i,j}}{\lambda}} 
$
and $ \lambda $ is the wavelength of the signal and $ d_{i,j} $ is the distance between antenna $ i $ and $ j $ at transmitter, relay or receiver.

The signal sent by the transmitter is given by
$
	\mt{\mSymbTransmitter} = \mAlpha{\mSymbTransmitter}\ms
$
where $ \ms $ are the transmitted symbols having an identity covariance matrix, $\mAlpha{\mSymbTransmitter} = \sqrt{\mPt}$, and $\mPt$ is the transmitted power. The signal received at the UAVs is given by
\begin{equation}
	\mr{\mSymbUAV} = \mHI \mt{\mSymbTransmitter} + \mn{\mSymbUAV}
\end{equation}
where $\mn{\mSymbUAV}\in\C{\mNu}$ is  additive white Gaussian noise with covariance $ \mnvar{\mSymbUAV} \b{I}$.  
The signal received by the UAVs is amplified and transmitted as 
$\mt{\mSymbUAV} = \mD \mr{\mSymbUAV} $,
where $ \mD $ is the amplification matrix which is defined as $ \mD = \text{diag}\{ \mAlpha{1},\cdots,\mAlpha{\mNu} \} $ where $ \mAlpha{i} $ is the gain of the $ i $th UAV. The signal received at the destination is
\begin{equation}
	\mr{\mSymbReceiver}=  \mAlpha{\mSymbTransmitter} \mHII \mD \mHI \ms + \mHII \mD  \mn{\mSymbUAV} + \mn{\mSymbReceiver}
\end{equation}
where $  \mn{\mSymbReceiver} $ is the additive white Gaussian noise with variance $ \mnvar{\mSymbReceiver} $. 

We assume that $ \mnvar{\mSymbReceiver} = \mnvar{} $ where $ \mnvar{} $ is the thermal noise and $ \mnvar{\mSymbUAV} = \mNF \mnvar{} $ where $ \mNF $ is the noise figure of the relay (the ratio between the input and output SNR). The noise covariance matrix of the end-to-end channel is $\mnVar{} = ( \mnvar{\mSymbUAV}  \mHII \mD  \mD^H \mHII^H + \mnvar{\mSymbReceiver}\b{I})$. The theoretical MIMO capacity of this end-to-end link is given by 
\begin{equation} \label{eq:capacity_general}
	C = \log_2(\det(\b{I} +  \mAlpha{\mSymbTransmitter}^2     \mHI^H  \mD^H \mHII^H   \mnVar{}^{-1} \mHII \mD \mHI  ))
\end{equation}

We consider two practical constraints on the relay amplifier. First, it has a maximum transmit power that it cannot exceed given by $ \mPMax{\mSymbUAV} $. Second, it has a maximum value of amplification gain given by $ \mAlphaMax{\mSymbUAV} $. This translates to the following constraints:	$\mAlpha{i}^2  |\left[\mr{\mSymbUAV} \right]_i|^2 \leq \mPMax{\mSymbUAV}$  and 
$ \mAlpha{i} \leq \mAlphaMax{\mSymbUAV} $.

The problem of interest is to find the placement of the UAVs which will maximize the capacity given by (\ref{eq:capacity_general}) while realizing the power and gain constraints. 
Optimizing power allocations for MIMO relays has been studied extensively in the literature (see \cite{relay_tutorial_2012} and the references within), so in this work we focus on UAV placement and use simple strategies to set the gains of the UAVs. Namely, we assume that all UAVs have the same gain $\mD = \mAlpha{U} \b{I}$,  which is of practical interest for simplified radio design in UAV relays and the reduced overhead of power allocation and control.


\section{Capacity Analysis}\label{sec:cap}
\subsection{Single UAV relay} \label{subsec:one_uav}
We will start with the simple case of $\mNt=\mNr=\mNu=1$ in order to gain insights about optimal placements of an UAV relay based on gain and power constraints. 
We can identify two regions, the first where the UAV is closer to the transmitter and is limited by the maximum power it can transmit and the second where it is limited by the maximum gain.
For the region of maximum gain, the expression for the capacity simplifies to
\begin{equation} \label{eq:single_uav}
C = \log_2 \left(1 +    \frac{\mAlpha{\mSymbTransmitter}^2 \mAlpha{\mSymbUAV}^2   \lambda^4        }{ \mnvar{} \mDistI^2 (\mNF \lambda^2 \mAlpha{\mSymbUAV}^2 (4\pi)^2   +  (4 \pi)^4  (\mDist-\mDistI)^2)} \right)   
\end{equation}
Solving for the maximas we get\footnote{Setting $ \mDistI $ to be equal to zero, while maximizes the expression, is not practical and violates the underlying assumptions of the LOS model.}
\begin{equation}\label{eq:roots_const_gain}
	\mDistI  = (1/4) ({3 \mDist+ \sqrt{\mDist^2-\frac{\lambda^2 \mAlpha{\mSymbUAV}^2 \mNF}{2\pi^2}}})
\end{equation}

As for the second solution $ \mDistI^{(2)} $, its existence in the feasible region $[0,\mDist]$ requires that relay gain $ \mAlpha{\mSymbUAV} $ and noise figure $ \mNF $ have low values. If these constants have a high value, i.e, the noise  figure is too high or the amplification (which also amplifies the noise) is too high, there will be no feasible second root  and the optimal capacity will be to move the relay closer to the transmitter.


Under maximum power constraint, we set $\mAlpha{\mSymbUAV}^2 = \frac{\mPMax{\mSymbUAV}}{\mAlpha{\mSymbTransmitter}^2|\mhI|^2}$, and capacity becomes 
\begin{equation}
C = \log_2 \left( 1 + \frac{ \mPMax{\mSymbUAV}\mAlpha{\mSymbTransmitter}^2 \lambda^2}{ (4\pi)^2 (\mPMax{\mSymbUAV} \mNF \mnvar{}  \mDistI^2 + \mAlpha{\mSymbTransmitter}^2  \mnvar{}  (\mDist - \mDistI)^2) }  \right)
\end{equation}
The maximum capacity is achieved at
\begin{equation}\label{eq:roots_max_power}
	\mDistI = {\mAlpha{\mSymbTransmitter}^2 \mDist}/({\mPMax{\mSymbUAV} \mNF + \mAlpha{\mSymbTransmitter}^2 })
\end{equation}
We can see from this expression that the optimal $\mDistI$ that maximizes capacity gets closer to the transmitter as we increase the noise figure of the UAVs or increase its maximum transmit power (which would lead to noise amplification). 
 \vspace{-0.5mm}
\subsection{UAV Swarm Relay}\label{subsec:multiple_uav}
\newtheorem{theorem}{Theorem}
In this section, we derive an upper bound for the capacity of the UAV swarm MIMO relay, then we propose a method to achieve that capacity. Assuming all UAVs have the same gain, the capacity becomes
\begin{equation}
C = \log_2(\det(\b{I} +    \mAlpha{\mSymbTransmitter}^2  \mAlpha{\mSymbUAV}^2 
 \mHI^H   \mHII^H  ( \mnVar{} )^{-1} \mHII  \mHI  ))
\end{equation}
where $ \mnVar{} = \mnvar{} \mHII  \mHII^H \mNF + \mnvar{}  \b{I}$.
Finding the position of UAV swarm and placement of UAVs within the swarm that maximize this capacity is a challenging problem in general. The magnitude and phase of each element of both $\mHI$ and $\mHII$ depend on distance between transmitter, UAV and receiver antennas.  We start our analysis by observing that this capacity is equivalent to the capacity of the channel $ \mHt = (   \mHII  \mHII^H \mNF + \b{I})^{-\frac{1}{2}} \mHII  \mHI $.  In the following, we derive the conditions on $\mHI$ and $\mHII$ that maximize the capacity.

\begin{theorem}
An upper bound on the capacity of the channel $\mHt$ is defined as $C\leq K\log(1+\frac{\mAlpha{\mSymbTransmitter}^2  \mAlpha{\mSymbUAV}^2 }{\mnvar{}K} \Tr(\mHt \mHt^{H}))$, where $K= \min(\mNt,\ \mNr,\ \mNu)$, and it is reached when $\mHt$ has orthogonal columns. 
\label{theorem:ubound1}
\end{theorem}

\begin{proof}
    The proof follows from the proof for a non-relay channel. The reader is referred to \cite[p.~295]{Tse_Wireless_2005} for the details of the latter proof.
\end{proof}
 \vspace{-0.5mm}
There are many candidates for the matrices $ \mHI $ and $ \mHII $, and effectively positions of UAVs,  that can realize this condition on $ \mHt $. From the expression for $ \mHt $, we can observe that one way to achieve this is to have $ \mHI $ have orthogonal columns and $ \mHII $ have orthogonal columns. We now derive a capacity equation for the case when this condition is satisfied. 

\begin{theorem} \label{thm2}
An upper bound on the capacity of the channel $\mHt$ is 
$C \leq K\log(1+\frac{\mAlpha{\mSymbTransmitter}^2  \mAlpha{\mSymbUAV}^2 }{\mnvar{}K} \sum_{i=1}^{K} \psi_{1,i}^{2}   \frac{\psi_{2,i}^{2}}{1 + \mNF\psi_{2,i}^{2}})$, 
where ${\psi}_{1,i}$ and ${\psi}_{2,i}$ are the singular values of $\mHI$ and $\mHII$ respectively, and $K = \min(\mNt,\ \mNr,\ \mNu)$. The upper bound is achieved when $\mHI$ and $\mHII$ are orthogonal.
\end{theorem}
\begin{proof}
	From the proof of Theorem \ref{theorem:ubound1} we know that 
	$	C \leq K\log(1+\frac{\mAlpha{\mSymbTransmitter}^2  \mAlpha{\mSymbUAV}^2 }{\mnvar{}K}\Tr(\mHt \mHt^{H}))$.
	Furthermore, we can expand the trace term as
	\begin{equation*}
	\Tr(\mHt \mHt^{H}) = \Tr((\mNF\mHII\mHII^{H}+I)^{-1}\mHII \mHI \mHI ^{H} \mHII^{H})
	\end{equation*}
	Let $\mHII = \boldsymbol{U}_{2}\boldsymbol{\Lambda}_{2}\boldsymbol{V}_{2}^{H}$ and $\mHI = \boldsymbol{U}_{1}\boldsymbol{\Lambda}_{1}\boldsymbol{V}_{1}^{H}$, by the singular value decomposition, then
	\begin{equation*}
	\Tr(\mHt \mHt^{H}) = \Tr((\mNF\boldsymbol{\Lambda}_{2}^{2}+I)^{-1}\boldsymbol{\Lambda}_{2}\mathbf{V}_{2}^{H}\boldsymbol{U}_{1}\boldsymbol{\Lambda}_{1}^{2}\mathbf{U}_{1}^{H}V_{2}\boldsymbol{\Lambda}_{2}^{T})
	\end{equation*}
	Using the fact that $\mathbf{U}_{1}^{H}\mathbf{V}_{2}$ is an orthogonal matrix, it can be shown that 
	\vspace{-3mm}
	\begin{equation}\label{eq:trace}
	\Tr(\mHt \mHt^{H}) \leq \sum_{i=1}^{K} \psi_{1,i}^{2}   \frac{\psi_{2,i}^{2}}{1 + \mNF\psi_{2,i}^{2}}
	\end{equation}
	and that this upper bound is satisfied in two cases. The first one is when $\mathbf{U}_{1}^{H}V_{2}$ is the identity matrix. The second case is when $\psi_{1,i}$ are all equal and $\psi_{2,i}$ are equal, i.e. $ \mHI $ and $ \mHII $ have orthogonal columns. Finally, using (\ref{eq:trace}) and Theorem \ref{theorem:ubound1} we can establish that
	\begin{equation}
	C \leq K\log_2 \left( 1+\frac{\mAlpha{\mSymbTransmitter}^2  \mAlpha{\mSymbUAV}^2 }{\mnvar{}K} \sum_{i=1}^{K} \psi_{1,i}^{2}   \frac{\psi_{2,i}^{2}}{1 + \mNF\psi_{2,i}^{2}} \right)
	\end{equation}
	and the bound is achieved when $\mHI $ and $ \mHII $ both have orthogonal columns.
\end{proof}

In the far field region, where the UAVs are at a large distance from the transmitter and the receiver compared to the relative size of the swarm, the magnitudes of channel coefficients in $\mHI $ and $ \mHII $ depend on $ \mDistI $ only. We can rewrite as $ \mHI \approx \frac{\lambda}{4\pi \mDistI} \mHbI $ and  $ \mHII \approx \frac{\lambda}{4\pi (\mDist - \mDistI)} \mHbII $, where $ |\left[\mHbI\right]_{i,j}| = |\left[\mHbII\right]_{i,j}| = 1$.  When $\mHI $ and $ \mHII $ both have orthogonal columns, i.e, $ \mHI^{H}\mHI = \mNt \mNu \mathbf{I}$ and $ \mHII^{H}\mHII = \mNu \mNr \mathbf{I}$, the capacity then becomes 
\begin{equation} \label{eq:capSwarm}
C \leq  K \log_2(1 +    \frac{\mAlpha{\mSymbTransmitter}^2 \mAlpha{\mSymbUAV}^2   \lambda^4   \phi_1^2 \phi_2^2     }{ \mnvar{} \mDistI^2 (\mNF \lambda^2 \mAlpha{\mSymbUAV}^2 (4\pi)^2    +  \mlHII^2  (4 \pi)^4  (\mDist-\mDistI)^2)} )   
\end{equation}
where $\phi_1 = \max(\mNt,\mNu)$ and $\phi_2=\max(\mNu,\mNr)$. The optimization becomes similar to the single UAV relay case given in (\ref{eq:single_uav}). Note that this results holds only in the far field and the power  received by each UAV and the destination antennas is almost equal. For small values of $\mDistI$ or values of $\mDistI$ close to $\mDist$, this relation is a loose upper bound.
\section{Proposed Method} \label{sec:method}
In this section, we propose a method to find positions that optimize the channel capacity for a UAV relay swarm. 
For a single UAV, to find the optimal value of $\mDistI$, we only need to evaluate the capacity at the roots derived in  (\ref{eq:roots_max_power}) and (\ref{eq:roots_const_gain}) within their respective regions and at the end points of the regions and then take the maximum point. 

\newcommand{\mSpacingTx}{ \ensuremath{d_{\mSymbTransmitter}} }
\newcommand{\mSpacingRx}{ \ensuremath{d_{\mSymbReceiver}} }
\newcommand{\mSpacingUAV}{ \ensuremath{d_{\mSymbUAV}} }
\newcommand{\mIntI}{ \ensuremath{m} } 
\newcommand{\mIntII}{ \ensuremath{n} } 
\newcommand{\mIndxDim}{ \ensuremath{i} } 

For the UAV swarm, the problem is more challenging since changing the position of one UAV affects both channel matrices.  Prior works have considered the problem of antenna array design for LoS MIMO between one transmitter and one receiver. Unfortunately, the geometry of the optimal receiver array depends on the geometry of the transmitter array. For this work, we assume that both the source and destination antennas in our problem consist of uniform rectangular arrays (URAs). The transmit antennas consists of $\mNt=\mNt^{(0)} \times \mNt^{(1)}$ antennas, where $\mNt^{(\mIndxDim)}$  where $\mIndxDim \in {0,1}$ corresponds to the dimension, having spacing $\mSpacingTx^{(\mIndxDim)}$. Similarly, the receivers consists of  $\mNr=\mNr^{(0)} \times \mNr^{(1)}$ and spacing  $\mSpacingRx^{(\mIndxDim)}$. 

If we only consider the first hop, the optimal capacity over this link occurs when the channel $\mHI$ is orthogonal. This can be achieved when the UAVs are placed in a parallel URA  having spacing \cite{bohagen_ura_2007}
\vspace{-1.5mm}
\begin{equation} \label{eq:orth1}
	\mSpacingTx^{(\mIndxDim)} \mSpacingUAV^{(\mIndxDim)} = \lambda \mDistI \left( \mIntI_\mIndxDim + 1 / \max{(\mNt^{(\mIndxDim)},\mNu^{(\mIndxDim)})}\right)
	\vspace{-1.5mm}
\end{equation}

for some integer $\mIntI_\mIndxDim$ under the condition that the distance $\mDistI$ is much larger than the dimension of the antennas, i.e, in the far field.
Similarly, for the second hop for some  integer $\mIntII_\mIndxDim$
\vspace{-1.5mm}
\begin{equation} \label{eq:orth2}
 \mSpacingUAV^{(\mIndxDim)} \mSpacingRx^{(\mIndxDim)} = \lambda (\mDist - \mDistI) \left(\mIntII_\mIndxDim + 1/\max{(\mNu^{(\mIndxDim)},\mNr^{(\mIndxDim)})}\right).
 \vspace{-1.5mm}
\end{equation}
If we find a value of $\mSpacingUAV$ for a given $\mDistI$, which satisfies both relations at the far field of both antennas, we can attain the capacity given in (\ref{eq:capSwarm}). An example where this relation applies is when  $\mDistI=\mDist/2$ and both the transmit and receive antennas have identical shapes ($\mNt^{(\mIndxDim)}=\mNr^{(\mIndxDim)}$, $\mSpacingTx^{(\mIndxDim)}=\mSpacingRx^{(\mIndxDim)}$), then both equations (\ref{eq:orth1}) and (\ref{eq:orth2}) can be trivially realized.  Unfortunately, such a value of  $\mSpacingUAV$ does not exist for all $\mDistI$. Additionally, the value of $\mDistI$ that maximizes the capacity given by  ($\ref{eq:capSwarm}$) depends on the value of gains and noise figure similar to the single UAV case. To address this we propose the following  simple search algorithm. 

We propose a two-step approach that first optimizes over $\mDistI$ and then finds spacing $\mSpacingUAV^{(\mIndxDim)}$ that maximizes the capacity. To find the optimal $\mDistI$, we use the approach for a single UAV placement optimization. Then, to a find UAV spacing, we use the fact that the capacity is optimized by improving both links. Therefore, we search the values of spacings between the optimal $\mSpacingUAV^{(\mIndxDim)}$ values of each link. This is done as follows: we use the relation (\ref{eq:orth1}) to calculate $\mSpacingUAV^{1(\mIndxDim)}$  and (\ref{eq:orth2}) to calculate $\mSpacingUAV^{2(\mIndxDim)}$. Then, for each of the dimensions, we scan the spacing of the UAVs between  $\mSpacingUAV^{2(\mIndxDim)}$ and  $\mSpacingUAV^{1(\mIndxDim)}$ and evaluate the capacity to find the spacing that results in a highest capacity.
 \vspace{-1.5mm}
\section{Simulations} \label{sec:simulations}
In this section, we aim to demonstrate the benefits of UAV placement optimization on the MIMO capacity and evaluate our proposed approach. The baseline scenario considered in the simulations assumes that the transmitter and receiver are $1$ km apart. They operate at carrier frequency of $5$ GHz over a narrowband channel. The transmitter power is set to $ \mPt = 12$ dBm, while the total power of the UAV swarm is $ \mPMax{\mSymbUAV} = 0 $ dBm and $ \mAlphaMax{\mSymbUAV}=45$ dB. The noise power corresponds to thermal noise with -174dBm/Hz over a bandwidth of 1 MHz. The UAVs are placed such that the lowest UAV has a height of 30 meters relative to the base of the transmit and receive antennas\footnote{As the height increases to values that are significant with respect to $\mDistI$, the assumption of having equal received power at the UAVs gets violated, which makes it not possible to achieve the upper bound.}.
\begin{figure}[t!]
	\centering
	\includegraphics{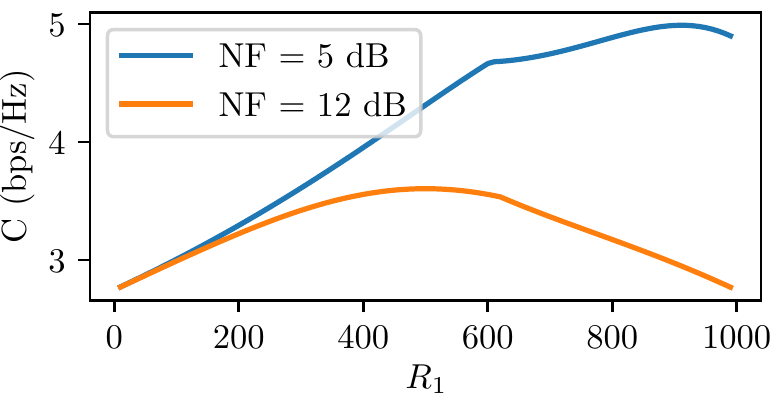}
	\vspace{-4mm}
	\caption{The capacity dependence on relay to transmitter distance for the single UAV case under maximum power constraint.}
	\label{fig:max_power}
\end{figure}

   \begin{figure}[t!]
	\centering
	\includegraphics{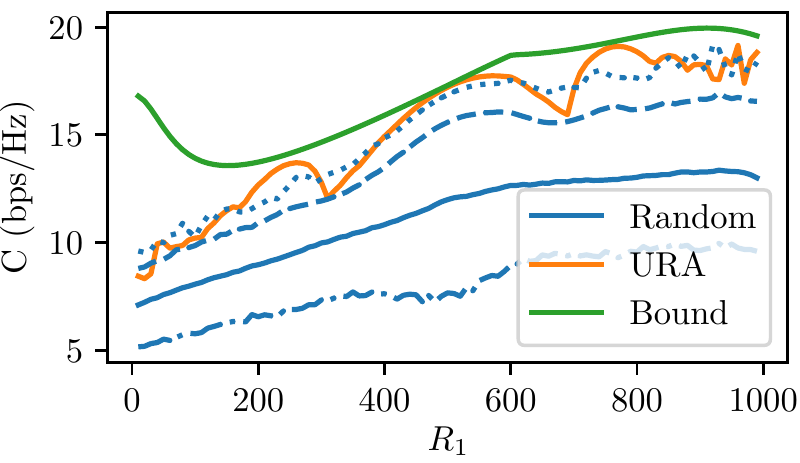}
	\vspace{-4mm}
	\caption{The  capacity obtained for 4 UAVs with NF=5dB. The green line is the theoretical upper bound. Dashed green line corresponds to near field. The orange line is the maximum capacity obtained using URA placement. The solid blue line shows the mean value of the achieved capacity using random placement. The dash-dotted and the dashed blue line represent the $5^{th}$ and $95^{th}$ percentile  capacities respectively. Dotted blue curve represents the maximum capacity obtained via random placement. }
	\label{fig:swarm_rand2}
	\vspace{-4mm}
\end{figure}
   \begin{figure}[t!]
	\centering
	\includegraphics{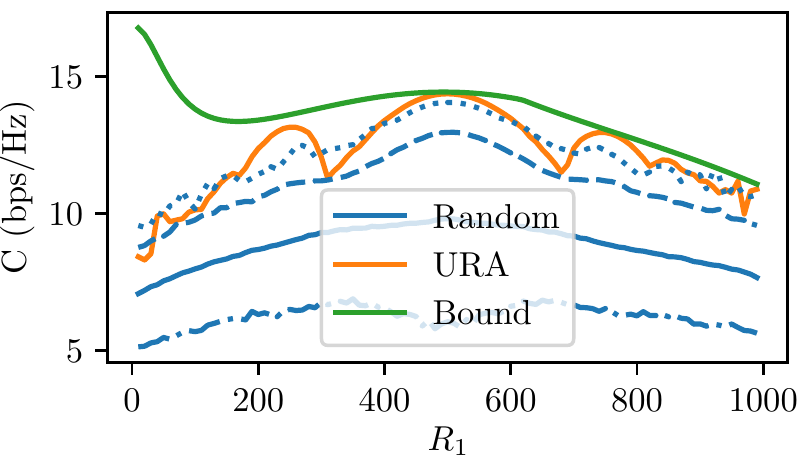}
	\vspace{-4mm}
	\caption{The  capacity obtained for 4 UAVs with NF=12dB. The color and line mapping is the same as in Figure. \ref{fig:swarm_rand2}}
	\label{fig:swarm_rand1}
	\vspace{-5mm}
\end{figure}


  \begin{figure}[t!]
 	\centering
 	\subfloat[Capacity for different UAV spacing. \label{fig:swarm_4_cap}]{\includegraphics[scale = 0.95]{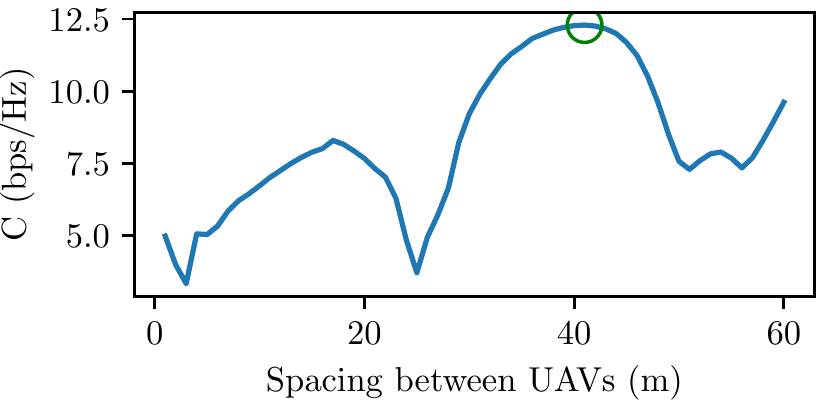}} \vspace{-2mm}
 	
 		\subfloat[Inverse condition number for different  UAV spacing. \label{fig:swarm_4_icn}]{\includegraphics{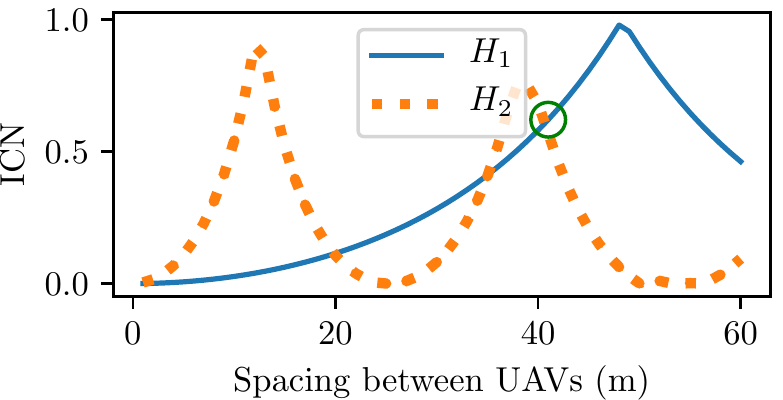}}
 	\caption{Capacity and ICN as  we change the spacing of $4\time 4$ UAV placed at $ \mDistI = 800$m. The optimal capacity and the corresponding ICN are highlighted with a green circle. }
 	\label{fig:swarm_4}
 	\vspace{-5mm}
 \end{figure}

We start by considering the single UAV relay scenario, where $\mNt=\mNr=\mNu=1$.
Fig. \ref{fig:max_power} shows the capacity of a UAV as it moves between $\mDistI=10$m and $\mDistI=990$m. In that case, we can identify two regions. In the first region, when the UAV is closer to the transmitter, the UAV needs to use a value of gain lower than the maximum gain in order to avoid exceeding the maximum power. As it moves further, it increases  $ \mAlpha{\mSymbUAV} $, thus compensating for the decay of $ \mhI $. In that case the maxima, if any, would be given by Eq.(\ref{eq:roots_max_power}). In the second region, the UAV reaches the maximum possible gain and it cannot further increase its amplification. In that case, the maxima, if exists within the region where the relation applies, would  be given by Eq. (\ref{eq:roots_const_gain}). 



For the UAV swarm relay, we consider the case where both the transmitter and receiver are $2\times2$ uniform rectangular arrays (URA) with 50 cm separation between antennas and a swarm consisting of 4 UAVs between the transmitter and receiver. Figures \ref{fig:swarm_rand2} and \ref{fig:swarm_rand1} show the capacities for different positions of the swarm and placements of UAVs within the swarm, assuming noise figures (NF) of 5 dB and 12 dB, respectively. In both figures, we show the upper bound given by (\ref{eq:capSwarm}), in green color. It is interesting to note that capacities for single UAV relay and UAV swarm MIMO relay follow the same trend in the far-field regions. We also note that for small $\mDistI$, the swarm is in near-field where assumptions used for deriving capacity bound (\ref{eq:capSwarm}) do not hold. 

Next, we evaluate whether random placement of UAVs, within a square area of 80m width, can achieve the capacity upper bound. Results in Figures \ref{fig:swarm_rand2} and \ref{fig:swarm_rand1} show maximum, average, $5^{th}$, and $95^{th}$ percentile over 5000 different placements at each $\mDistI$.
These results show that there exists random placement that achieves the upper bound in the region where the UAVs are in the far-field of both antennas. For small values of $\mDistI$, when swarm is close to the transmitter, the gap between the upper bound and the obtained values is large. On the other hand, for large $\mDistI$, the noise covariance in (\ref{eq:capacity_general}) becomes $\mnVar{} \approx  (\mHII \mHII^H)^{-1}$ and the effect of $\mHII$ cancels out, while conditions for far-field apply to $\mHI$. As a result, in the region where swarm is closer to the receiver, the gap between achievable capacity and upper bound is smaller.  

We also evaluate whether proposed URA placement of UAVs can achieve the capacity upper bound. The results in Figures   \ref{fig:swarm_rand2} and \ref{fig:swarm_rand1} show that there exists optimal URA placement that attains the upper bound around mid point, $\mDistI=\mDist/2$, as well as several values of $\mDistI$, where both channels can be made orthogonal. Although the proposed UAV placement based on URA geometry does not always attain the upper bound, it can always achieve a capacity better than 95\% of random placements.

In Fig. \ref{fig:swarm_4_cap} and \ref{fig:swarm_4_icn} we investigate the orthogonality conditions of $\mHI$ and $\mHII$ for different URA spacings. Fig.  \ref{fig:swarm_4_cap} shows the capacity and Fig. \ref{fig:swarm_4_icn} shows the inverse condition number (ICN), ratio between the smallest and largest singular values, for both the matrices $\mHI$ and $\mHII$. The ICN is equal to 1 when a  matrix has orthogonal columns. From Fig. \ref{fig:swarm_4_icn}, we can see that $\mHII$ becomes  orthogonal for two values of $\mSpacingUAV$ within that range (corresponding to two different values of $\mIntII$ in (\ref{eq:orth2})) and $\mHI$ becomes  orthogonal at only one value. By sweeping the spacing in the range of $\mSpacingUAV$'s that orthogonalize each channel, the algorithm can find the spacing that maximizes capacity.

\section{Conclusion and Future Work} \label{sec:conclusion}

In this paper, we studied the capacity improvements of an obstructed MIMO link enabled by a UAV swarm acting as an amplify-and-forward MIMO relay under a line-of-sight channel. Our analysis revealed that significant gains can be achieved by optimizing the placement of the position of the entire swarm as well as the UAVs within the swarm.  We derived an upper bound of the capacity for UAV swarm. An algorithm that approaches this upper bound is proposed for the URA transmitter and receiver cases.  Further research is needed to find a tighter upper bound for a planar UAV swarm placement and algorithms to achieve it  when the UAVs are closer to the transmitter or receiver. Additionally, joint optimization of power gains and locations is a promising approach to further improve capacity.


\bibliography{references}
\bibliographystyle{ieeetr}

\end{document}